\documentclass[a4paper,11pt]{article}
\usepackage{amsmath,amsfonts,amssymb,amsthm}
\usepackage{mathrsfs}
\usepackage{bm}
\usepackage{thm-restate}
\usepackage{enumitem}
\usepackage{tikz-cd}
\usetikzlibrary{arrows.meta}
\usetikzlibrary{matrix,arrows,decorations.pathmorphing}
\usetikzlibrary{patterns}
\usepackage[utf8]{inputenc}
\usepackage{a4wide}
\usepackage{complexity}
\usepackage{tikz}
\usepackage{xurl}
\usepackage{hyperref}
\hypersetup{colorlinks=true,citecolor=blue,linkcolor=blue,urlcolor=blue}
\usepackage{microtype}
\usepackage{comment}

\usepackage{tikz-cd}
\usetikzlibrary{arrows.meta}
\usetikzlibrary{matrix,arrows,decorations.pathmorphing}
\usepackage[utf8]{inputenc}
\usepackage{complexity}
\usepackage{tikz}
\usetikzlibrary{patterns}
\pgfdeclarelayer{bg}
\pgfsetlayers{bg, main}
\usetikzlibrary{fit,matrix}

\DeclareMathOperator{\lvl}{lvl}
\DeclareMathOperator{\ar}{ar}

\pgfdeclarelayer{bg}
\pgfsetlayers{bg, main}
\usetikzlibrary{fit,matrix}

\newcommand{\conn}{\ensuremath{\simeq}}
\newcommand{\NAE}{\ensuremath{\mathbf{NAE}}}
\newcommand{\Pol}{\ensuremath{\mathrm{Pol}}}
\newcommand{\PCSP}{\ensuremath{\mathrm{PCSP}}}
\newcommand{\CSP}{\ensuremath{\mathrm{CSP}}}
\newcommand{\inn}[2]{\ensuremath{\mathbf{#1}\textbf{-}\mathbf{in}\textbf{-}\mathbf{#2}}}

\renewcommand{\A}{\mathbf{A}}
\newcommand{\B}{\mathbf{B}}
\newcommand{\X}{\mathbf{X}}
\renewcommand{\C}{\mathbf{C}}
\renewcommand{\D}{\mathbf{D}}
\renewcommand{\G}{\mathbf{G}}
\renewcommand{\H}{\mathbf{H}}
\renewcommand{\K}{\mathbf{K}}
\renewcommand{\L}{\mathbf{L}}
\renewcommand{\P}{\mathbf{P}}
\newcommand{\Q}{\mathbf{Q}}
\newcommand{\T}{\mathbf{T}}

\newcommand{\bx}{\mathbf{x}}

\newcommand{\Z}{{\mathbf{Z}}}

\newcommand{\HG}{\widehat{\mathbf{G}}}
\newcommand{\HD}{\widehat{\mathbf{D}}}
\newcommand{\HH}{\widehat{\mathbf{H}}}

\newcommand{\HT}{\widehat{\mathbf{T}}}
\newcommand{\minion}[1]{{\ensuremath{\mathscr{#1}}}}

\newcommand{\zaff}{\mathbf{Z}}

\renewcommand\phi\varphi

\let\rel\mathbf

\setlist[enumerate]{label={(\roman*)}, ref={\roman*}}

\theoremstyle{plain}
\newtheorem{theorem}{Theorem}
\newtheorem{lemma}[theorem]{Lemma}
\newtheorem*{lemma*}{Lemma}
\newtheorem{proposition}[theorem]{Proposition}
\newtheorem*{proposition*}{Proposition}
\newtheorem{corollary}[theorem]{Corollary}
\newtheorem*{corollary*}{Corollary}

\newtheorem{problem}[theorem]{Problem}
\theoremstyle{definition}
\newtheorem{definition}[theorem]{Definition}
\newtheorem{remark}[theorem]{Remark}

\begin{document}

\title{\texorpdfstring{\textsc{1-in-3} vs.~\textsc{Not-All-Equal}:\\ Dichotomy of a broken promise}{1-in-3 vs. Not-All-Equal: Dichotomy of a broken promise}\thanks{An extended abstract of part of this work has appeared in the Proceedings of LICS 2024~\cite{ckknz24:lics}. This work was supported by EPSRC Fellowship EP/X033201/1, UKRI EP/X024431/1, and a Clarendon Fund Scholarship. Partially funded by the National Science Centre, Poland under the  Weave-UNISONO  call  in  the  Weave  programme 2021/03/Y/ST6/00171. For the purpose of Open Access, the authors have applied a Creative Commons Attribution (CC BY) license to any Accepted Manuscript version arising. All data is provided in full in the results section of this paper.}}

\author{Lorenzo Ciardo\\
University of Oxford\\
\texttt{lorenzo.ciardo@cs.ox.ac.uk}
\and
Marcin Kozik\\
Jagiellonian University\\
\texttt{marcin.kozik@uj.edu.pl}
\and
Andrei Krokhin\\
Durham University\\
\texttt{andrei.krokhin@durham.ac.uk}
\and
Tamio-Vesa Nakajima\\
University of Oxford\\
\texttt{tamio-vesa.nakajima@cs.ox.ac.uk}
\and
Stanislav \v{Z}ivn\'y\\
University of Oxford\\
\texttt{standa.zivny@cs.ox.ac.uk}
}

\maketitle

\begin{abstract}
The \textsc{1-in-3} and \textsc{Not-All-Equal} satisfiability problems for Boolean CNF formulas are two well-known \NP-hard problems. In contrast, the \emph{promise}  \textsc{1-in-3} vs.~\textsc{Not-All-Equal} problem can be solved in polynomial time. 
In the present work, we investigate this constraint satisfaction problem in a regime where the promise is weakened from either side by a \emph{rainbow-free} structure, and establish a complexity dichotomy for the resulting class of computational problems.
\end{abstract}

\section{Introduction}
\label{sec:intro}

Let $\phi$ be a Boolean formula given as a conjunction of clauses, each consisting of three (un-negated) variables. Consider the following question:
\begin{center}
    \emph{Is there a truth assignment such that\\
    each clause has exactly one true variable?}
\end{center} 
This is a well-known \NP-hard computational problem, called (monotone) \textsc{1-in-3 Sat}~\cite{Garey79}. Similarly, the question 
\begin{center}
    \emph{Is there a truth assignment such that\\
    each clause has at least one true and one false variable?}
\end{center}
is the \NP-hard problem known in the literature as \textsc{3-Not-All-Equal (NAE) Sat}~\cite{Schaefer78:stoc,Garey79}.
In these and other variants,
the Boolean satisfiability problem has had a central importance in the development of complexity theory, its investigation dating back at least to~\cite{Cook71}. 
Notice that the first notion of satisfiability is stronger than the second: Any \textsc{1-in-3} assignment is also an \textsc{NAE} assignment. Consider now the \emph{promise} satisfiability problem that asks to distinguish whether a formula $\phi$ is satisfiable in the strong sense (a \textsc{1-in-3} assignment exists) or $\phi$ is not even satisfiable in the weak sense (an \textsc{NAE} assignment does not exist).
This problem --- known as ``\textsc{1-in-3} vs.~\textsc{NAE}''~\cite{BG21:sicomp} --- is a relaxation of both problems considered above, in that it admits any answer on those formulas that are satisfiable in the weak but not in the strong sense. Equivalently, one is promised that the input formula is not of that kind.
Let us try to solve this problem, using an algorithm from~\cite{BG21:sicomp}. For any clause in $\phi$ involving three variables $x,y,z$, consider the linear equation 
\begin{equation}
\label{eqn_xyz}
    x+y+z=1.
\end{equation}
This results in a linear system, which may be solved over the integers in polynomial time by using, essentially, Gaussian elimination.\footnote{More precisely, by the integral version of Gaussian elimination that corresponds, in matrix terms, to computing the Hermite
or Smith normal forms of the matrix of the linear system~\cite{KB79}.} If there is no integer solution, we are sure that, in particular, no $\{0,1\}$ solution exists: $\phi$ does not admit a $\textsc{1-in-3}$ assignment. If there is an integer solution, we round it by turning positive values into $1$ and non-positive values into $0$.
Since no three positive (respectively, non-positive) integers can sum up to $1$, we are guaranteed that the output of this process is a valid \textsc{NAE} assignment --- while it is not necessarily a \textsc{1-in-3} assignment, as is witnessed, for example, by the solution to~\eqref{eqn_xyz} given by $x=y=2$, $z=-3$.

In other words, while \textsc{1-in-3 Sat} is \NP-hard, if we are promised that all of the unsatisfiable formulas we are considering are not even satisfiable in the weaker \textsc{NAE} sense, the problem becomes tractable (solvable in polynomial time). Similarly (and dually),
the promise that all \textsc{NAE}-satisfiable formulas are also \textsc{1-in-3}-satisfiable turns \textsc{NAE} into a tractable problem.
It is then natural to investigate what happens if we modify the promise. Clearly, a stronger promise would lead to an even easier problem and, in particular, to a tractable one. What if we weaken it? How does the promise impact on the complexity behaviour of the problem? Where is the boundary of tractability?

In order to formulate these questions in a formal way, it shall be convenient to use the paradigm of \emph{Constraint Satisfaction Problems} (CSPs), which provides a broader context for capturing Boolean satisfiability problems, as well as other computational problems such as graph and hypergraph colouring. We can phrase a CSP as a homomorphism problem, where the objective is to test for the existence of a homomorphism between an \emph{instance} structure $\X$ and a \emph{template} structure $\A$. In the setting of satisfiability of Boolean formulas, we should think of $\X$ as a proxy for the formula $\phi$, while $\A$ encodes the satisfiability notion we are considering. In this formulation, $\X$ and $\A$ are two similar (finite) \emph{relational structures}, consisting of finite \emph{domains} ($X$ and $A$, respectively), as well as \emph{relations} ($R^\X\subseteq X^r$ and $R^\A\subseteq A^r$, respectively) for each \emph{relation symbol} $R$, where the positive integer $r$ is the \emph{arity} of $R$. A homomorphism between $\X$ and $\A$ is a map $f:X\to A$ that preserves the relations; i.e., $f(\bx)\in R^\A$ whenever $\bx\in R^\X$, where $f$ is applied entrywise. We denote the existence of a homomorphism by $\X\to\A$. The CSP parameterised by $\A$, denoted by $\CSP(\A)$, is the computational problem: ``Given an instance $\X$, output \textsc{Yes} if $\X\to\A$ and \textsc{No} if $\X\not\to\A$''. 
If we define the Boolean structure $\inn{1}{3}$ whose unique relation, of arity $3$, is the set 
\[
    \{(0,0,1),(0,1,0),(1,0,0)\},
\]
then
$\CSP(\inn{1}{3})$ is precisely the \textsc{1-in-3 Sat} problem. Similarly, we can formulate the \textsc{NAE Sat} problem as the CSP parameterised by the Boolean structure $\NAE$ whose unique relation, of arity $3$, is the set
\[
    \left\{
    \begin{array}{l}
    (0,0,1),(0,1,0),(1,0,0),\\(1,1,0),(1,0,1),(0,1,1)
    \end{array}
    \right\}.
\]
Other classic examples of CSPs are homomorphisms problems for digraphs (which are relational structures having a single, binary relation) and, more generally, hypergraphs. In particular, the CSP parameterised by the $n$-clique $\K_n$ is the well-known graph $n$-colouring problem.

Several decades of research efforts
have equipped the framework of CSPs with a rather sharp set of tools --- mostly coming from universal algebra --- that can be leveraged to explain the computational complexity of satisfiability problems. More precisely, the complexity of $\CSP(\A)$ is entirely determined by a certain type of identities holding in the \emph{polymorphism clone} of $\A$, which contains all homomorphisms of the form $\A^k\to\A$ (where $\A^k$ is the $k$-fold direct power of $\A$)~\cite{Jeavons97:closure,Jeavons98:algebraic,Bulatov05:classifying,LaroseT09,BOP18,BKW17}.

This correspondence between polymorphisms and complexity 
--- ultimately based on a  \emph{Galois connection} between sets of relations and operations --- 
has been invaluable in the exploration of the complexity landscape of CSPs. Eventually, it has led to the positive resolution of Feder-Vardi's Dichotomy Conjecture (now Theorem) by Zhuk~\cite{Zhu20} and Bulatov~\cite{Bul17}, which asserts that a CSP is tractable in polynomial time if it has a polymorphism satisfying an identity of a certain kind, and it is \NP-hard otherwise~\cite{FV98}. 

Promise problems like \textsc{1-in-3} vs.~\textsc{NAE} are captured by a paradigm, known as \emph{Promise $\CSP$s} (PCSPs for short), that generalises CSPs. Here, the template is a \emph{pair} $(\A,\B)$ of (finite) structures, and the computational problem $\PCSP(\A,\B)$ is ``Given an instance $\X$, output \textsc{Yes} if $\X\to\A$ and \textsc{No} if $\X\not\to\B$''. In order for the \textsc{Yes} and the \textsc{No} instances to be disjoint, we require that $\A\to\B$. PCSPs were introduced in~\cite{AGH17,BG21:sicomp} to unify the study of
approximability of perfectly satisfiable CSPs.
The PCSP framework vastly extends CSPs: Firstly,
several well-known computational problems can be formulated in the former, but not in the latter. Primary examples include the approximate colouring problems, which we shall discuss in more detail later. Secondly, already the early exploration of PCSPs unveiled a number of new phenomena --- absent in the non-promise setting --- that quickly called for a more general and conceptually different approach to their study, going beyond the universal-algebraic approach to CSPs. The crux of this need lies in the fact that the Galois connection for PCSPs is far less structured than the one for CSPs: While the complexity of PCSPs is still governed by polymorphisms (which are now homomorphisms $\A^k\to\B$), the algebraic structure that they form does not admit composition in the promise context. A consequence of this fact is that the universal-algebraic tools that allow generating an infinite set of new identities from a single polymorphic identity fail for PCSPs. This, in turn, stimulated the use of different tools to study PCSP polymorphisms, including Boolean function analysis~\cite{BGS23}, topology~\cite{KOWZ22}, matrix and tensor theory~\cite{cz23sicomp:clap,cz23soda:aip,cz23stoc:ba}, and Fourier analysis~\cite{HechtMS23}.

It is, of course, a very natural question whether the CSP dichotomy extends to PCSPs. 
Before being able to even conjecture a dichotomy for such a wide class of problems, it would be beneficial to obtain classifications in well-chosen special cases. 
When the Feder-Vardi conjecture was made, it was supported, in particular, by important special cases that were classified at the time:
the Boolean CSP (i.e.,~the domain of $\A$ is $\{0,1\}$) \cite{Schaefer78:stoc}, the afore-mentioned graph colouring problem~\cite{Karp72}, and the undirected graph homomorphism problem (i.e.,~$\A$ is an undirected graph)~\cite{HN90}. (Note that the second problem is a special case of the third.)
The promise versions of the Boolean CSP, graph colouring, and graph homomorphism problem
have been studied and there are some partial complexity classification results
about them (see, e.g.~\cite{AGH17,BG21:sicomp,BGS23,Ficak19:icalp,BBKO21,KOWZ22}), but
full classifications even in these cases are believed to be difficult to obtain. In particular, the promise version of graph colouring is the well-known \emph{approximate graph colouring} problem: checking whether a given graph is $n$-colourable or not even $m$-colourable, for given $n\leq m$. The complexity of this problem, that corresponds to $\PCSP(\K_n,\K_m)$, is a long-standing open question in computer science~\cite{GJ76}. It was only resolved in certain special cases~\cite{Khanna00:combinatorica,BrakensiekG16,BBKO21,Huang13,KOWZ22} or under strong complexity-theoretic assumptions such as variants of the Unique Games Conjecture~\cite{DMR09,GS20,Braverman21:focs}. 
Other particular examples of PCSPs have been studied in~\cite{Barto21:stacs,Brandts22:phd,bz22:ic,nz23,NVWZ_CF25}. The PCSP templates considered there have either small domains or specific structure.

The \textsc{1-in-3} vs.~\textsc{NAE} problem --- i.e.~$\PCSP(\inn{1}{3},\NAE)$ --- is in many respects a prototypical example of PCSP, in that it witnesses several of the new behaviours separating the promise and the non-promise worlds. Next, we discuss three of these separating behaviours.

\begin{enumerate}
\item If three structures $\A,\B,\C$ are such that $\A\to\C\to\B$, then
  $\PCSP(\A,\B)$ reduces to $\CSP(\C)$ through the trivial reduction that does
    not change the instance.\footnote{In the PCSP literature, this situation is
    known as a \emph{sandwich}~\cite{BG21:sicomp}.} Hence, if $\CSP(\C)$ is tractable, the exact same algorithm solving it also solves $\PCSP(\A,\B)$. The tractability of many PCSPs can be certified through this pattern, by exhibiting a suitable structure $\C$. For \emph{finite} $\C$, this is provably not the case for $\PCSP(\inn{1}{3},\NAE)$: It was established in~\cite{BBKO21} that any finite structure $\C$ for which $\inn{1}{3}\to\C\to\NAE$ is such that $\CSP(\C)$ is \NP-hard.\footnote{\label{footnote_barto_only_holds_if_finite}We remark that the statement is false if we admit structures having an \emph{infinite} domain. In fact, the algorithm based on Gaussian elimination that was described at the beginning of the Introduction can be reformulated as a reduction of $\PCSP(\inn{1}{3},\NAE)$ to $\CSP(\zaff)$, where $\zaff=(\mathbb{Z};\{(x,y,z)\in\mathbb{Z}^3\mid x+y+z=1\})$ satisfies $\inn{1}{3}\to\zaff\to\NAE$.} This means that for $\PCSP(\inn{1}{3},\NAE)$ the source of its tractability lies outside of the scope of non-promise finite CSPs and comes from \emph{infinite-domain} CSPs. (We note that there is no dichotomy for infinite-domain CSPs~\cite{Bodirsky08:icalp}, although there is a conjecture that the dichotomy for finite-domain CSPs extends to a certain well-behaved class of infinite-domain CSPs, cf.~\cite{BPP21}.)

\item \emph{Local consistency} is an algorithmic technique that consists in relaxing the question ``Does a homomorphism $\X\to\A$ exist?'' by
testing whether all subinstances of $\X$ of bounded size admit a system of compatible homomorphisms to $\A$.
This method applies to both CSPs and PCSPs; the templates that can be solved by enforcing local consistency are said to have a \emph{bounded width}. In the realm of CSPs, local consistency has precisely the same power as the linear-programming based Sherali-Adams hierarchy~\cite{tz17:sicomp}. However, it was recently shown in~\cite{AD22} that $\PCSP(\inn{1}{3},\NAE)$ is solved by some fixed round of the Sherali-Adams hierarchy and yet it has unbounded width, thus implying a lack of collapse of the two algorithmic models for PCSPs.

\item 
Another singular behaviour of $\PCSP(\inn{1}{3},\NAE)$ emerged in the context of robust algorithms. The class of CSPs that can be solved \emph{robustly} --- i.e., through some algorithm that is not highly sensitive to small noise in the instance --- coincides with the class of bounded-width problems~\cite{Barto16:sicomp}. Nevertheless, it was established in~\cite{bgs_robust23stoc} that the problem $\PCSP(\inn{1}{3},\NAE)$ can be solved robustly although it is of unbounded width --- thus yielding another discrepancy between the promise and the non-promise settings.
\end{enumerate}

For these reasons, a deeper understanding of \emph{why} problems akin to $\PCSP(\inn{1}{3},\NAE)$ are tractable can shed light on the new type of behaviours we expect to find in the complexity landscape of promise problems. 
To that end, in this paper we investigate the tractability boundary of problems that weaken the promise of $\PCSP(\inn{1}{3},\NAE)$ either from the left or from the right --- when the weakened version of the promise continues to share key features with $\inn{1}{3}$ and $\NAE$ (namely being \emph{symmetric} and \emph{rainbow-free}), we find a dichotomy for the corresponding fragments of PCSPs. We note that~\cite{bz22:ic}, motivated by the same goal, considered a certain class of Boolean templates that result from weakening the promise of one of the two structures. Unlike~\cite{bz22:ic}, we go beyond Boolean domains.

The relational structures $\inn{1}{3}$ and $\NAE$ contain one ternary relation, which is \emph{symmetric}, i.e., it is invariant under permutations of the arguments, and it is \emph{rainbow-free}, i.e., it does not contain any tuple $(x,y,z)$ whose arguments are all distinct (note that this is always the case for Boolean ternary structures).
Following~\cite{Barto21:stacs}, we describe this type of relational structures by associating digraphs to them. More precisely, given a digraph $\G=(V,E)$, we let $\HG$ be the relational structure
defined by
\[
\HG = (V; \{ (x, x, y), (x, y, x), (y, x, x) \mid (x, y) \in E \}).
\]
For example, if $\rel G$ is a single directed edge (from 0 to 1) or a single undirected edge, the corresponding structure $\HG$ is $\inn{1}{3}$ or $\NAE$, respectively. It is easy to check that we have $\G\to\H$ if and only if $\HG\to\HH$. Moreover, $\PCSP(\G,\H)$
always reduces to $\PCSP(\HG,\HH)$, but the latter problem can be harder --- for
example, if $\G=\H$ is an undirected edge, $\PCSP(\G,\H)=\CSP(\G)$ is the
(tractable) 2-colouring problem, whereas
$\PCSP(\HG,\HH)=\CSP(\HG)=\CSP(\NAE)$ is \NP-hard.

We are now ready to state our main results, which concern the regime where the promise of the \textsc{1-in-3} vs.~\textsc{NAE} problem is ``broken'' from either side --- i.e., 
one of the two structures in the template $(\inn{1}{3},\NAE)$ is replaced by a different structure.\footnote{We remark that PCSPs (and CSPs) can also be formulated in their \emph{search} version, as opposed to the \emph{decision} versions discussed in this Introduction. The two versions are polynomial-time equivalent for CSPs~\cite{Bulatov05:classifying}; for PCSPs, decision reduces to search, but it is not known whether an efficient reduction in the other direction always exists. Our results hold for both decision and search versions of PCSPs.} 
First, fix $\inn{1}{3}$ and consider any digraph $\G$ such that $(\inn{1}{3},\HG)$ is a valid template --- which happens if and only if $\inn{1}{3}\to\HG$ or, equivalently, if and only if $\G$ contains at least one edge.
\begin{restatable}{theorem}{mainone}
\label{thm:main1}
     $\PCSP(\inn{1}{3}, \HG)$ is tractable if $\G$ has a directed cycle of length at most 3, and it is \NP-hard otherwise.
\end{restatable}

Second, fix $\NAE$ and consider any digraph $\G$ such that $(\HG,\NAE)$ is a valid template --- which happens if and only if $\HG\to\NAE$ or, equivalently, if and only if the graph obtained from 
$\G$ by forgetting directions is bipartite. 
Following \cite{HN04}, we say that a digraph is \emph{balanced} if each of its
oriented cycles has as many edges in one direction as in the other.
\begin{restatable}{theorem}{maintwo}
\label{thm:main2}
    $\PCSP(\HG, \NAE)$ is tractable if $\G$ is balanced, and it is \NP-hard otherwise.
\end{restatable}

\paragraph{Related work.} We now discuss prior work on the problem $\PCSP(\inn{1}{3}, \A)$.
All our tractable cases can be solved by the Affine Integer Programming (AIP) relaxation of~\cite{BG21:sicomp}. This provides more evidence for the conjecture, first formally stated in the work~\cite{NZ24:talg}, that
$\PCSP(\inn{1}{3}, \A)$ is tractable if and only if it is solved by AIP.

Given a digraph $\G$, let $\HG^+$ denote the ternary structure obtained from $\HG$ by adding all rainbow tuples (i.e., tuples $(x,y,z)$ with $x,y,z$ all distinct) to the relation. Let $\T_i$ denote the transitive tournament with $i$ vertices.~\cite{Barto21:stacs} classified the complexity of all problems $\PCSP(\inn{1}{3}, \A)$ where $\A$ has a 3-element domain, with the exception of $\PCSP(\inn{1}{3}, \HT_3^+)$.\footnote{In~\cite{Barto21:stacs}, $\HT_3$ and $\HT_3^+$ are denoted by
$\mathbf{T}_1$
and $\mathbf{T}_1^+$, respectively.} 
Note that $\inn{1}{3}$ is precisely $\HT_2^+$. In
\cite{Barto21:stacs}, it is conjectured that 
$\PCSP(\HT_k^+, \HT_\ell^+)$
is \NP-hard for $k \leq \ell$.\footnote{In~\cite{Barto21:stacs}, $\HT_k^+$ is denoted by $\mathbf{LO}_k$.}
Versions of this problem for higher arities were shown to be hard in~\cite{nz23,NVWZ_CF25}, and the techniques of~\cite[Theorem 26]{nz23} show that to prove this conjecture it is necessary and sufficient to show that
\[
\PCSP(\HT_2^+, \HT_k^+) = \PCSP(\inn{1}{3}, \HT_k^+)
\]
is \NP-hard for all $k \geq 2$. In~\cite{FNOTW24}, it was proved that $\PCSP(\HT_3^+, \HT_4^+)$ is \NP-hard using topological methods. We view our hardness results as a partial step towards the resolution of this conjecture, in the rainbow-free regime: It immediately 
follows from Theorem~\ref{thm:main1} that $\PCSP(\inn{1}{3}, \HT_k)$ is \NP-hard for $k \geq 2$; what the conjecture then requires is for hardness to hold even if rainbow tuples are allowed.

Nakajima and \v{Z}ivn\'y~\cite{NZ24:talg} classified all PCSPs which have the form $\PCSP(\inn{1}{3}, \A)$ where $\A$ is symmetric and 
\emph{functional}, that is, the relation of $\A$ does not contain tuples $(x, y, z), (x, y, z')$ with $z \neq z'$. Thus, these results fail to classify $\PCSP(\inn{1}{3}, \HG)$ whenever $\G$ has a vertex with out-degree at least 2. On the other hand, unlike this paper, some of the structures $\A$ for which~\cite{NZ24:talg} gave a classification do have rainbow tuples.

Comparing with the literature on $\PCSP(\inn{1}{3}, \A)$, we find it interesting that the dual extension $\PCSP(\A, \NAE)$ has been significantly less investigated. We are aware of only two papers in this line of work.
Firstly, \cite{bz22:ic} studied $\PCSP(\A,\NAE)$ for Boolean (possibly non-symmetric) $\A$ obtained from $\inn{1}{3}$ by adding extra tuples.
Secondly, 
\cite{Ficak19:icalp}
established
a classification of Boolean symmetric PCSPs  --- however, up to isomorphism, the only Boolean symmetric relational structures with one ternary relation that map to $\NAE$ are $(\{0, 1\}; \emptyset)$, $\inn{1}{3}$, and $\NAE$. Thus, for the case $\PCSP(\HG, \NAE)$ with $\HG$ Boolean, the results of~\cite{Ficak19:icalp} only yield trivial results: The problem is tractable when $\G$ is empty or has a loop, or when $\HG = \inn{1}{3}$, and is otherwise \NP-complete. Our results thus additionally cover all other digraphs $\G$ --- i.e., the non-Boolean setting.\\

The rest of the paper is dedicated to the proofs of Theorems~\ref{thm:main1} and~\ref{thm:main2}, which are given in Sections~\ref{sec_first_theorem} and~\ref{sec_second_theorem}, respectively, after introducing some preliminary notions in Section~\ref{sec:prelims}. In Section~\ref{sec_conclusion}, we shall outline some natural directions for future investigation.

\paragraph{Acknowledgements} We thank all anonymous reviewers of various versions
of this paper for their comments and suggestions for changes.

\section{Preliminaries}
\label{sec:prelims}
We denote by $[n]$ the set $\{1,2,\ldots,n\}$. Moreover, we denote by $[n, m)$ the set $\{n, \ldots, m - 1\}$. We say that a collection of sets $A_1, \ldots, A_n$ is disjoint if $A_i \cap A_j = \emptyset$ whenever $i \neq j$. We consider that $\mathbb{N} = \{0, 1, \ldots, \}$, i.e.~the natural numbers include 0.

\paragraph*{Relational structures and homomorphisms}
Except for digraphs, all relational structures that appear in this paper 
are pairs $\A=(A;R^\A)$, where $A$ is the domain and $R^\A\subseteq A^3$ is a
ternary relation. 
A relational structure $\A$ is called \emph{symmetric} if the relation
$R^\A$ is
symmetric, i.e., invariant under any permutation of its three arguments. 

Given two relational structures $\A=(A;R^\A)$ and $\B=(B;R^\B)$, a
\emph{homomorphism} from $\A$ to $\B$ is a function $h:A\to B$ such that
$(h(x),h(y),h(z))\in R^\B$ whenever $(x,y,z)\in R^\A$. We denote the existence
of a homomorphism from $\A$ to $\B$ by $\A\to\B$. A pair of relational
structures $(\A,\B)$ with $\A\to\B$ is called a (PCSP) \emph{template}.

\paragraph*{Polymorphisms, minions}
We say that a function $f:A^n\to B$ has \emph{arity} $\ar(f) = n$.
For such a function, we say that a coordinate $i \in [n]$ is \emph{essential} if there exist $a_1, \ldots, a_n, a_i' \in A$ such that
\begin{align*}
\begin{array}{c}
     f(a_1, \ldots, a_{i-1}, a_i, a_{i+1}, \ldots, a_n)\\
     \neq\\
     f(a_1, \ldots, a_{i-1}, a_i', a_{i+1}, \ldots, a_n).
\end{array}
\end{align*}
An $n$-ary function $f$ has \emph{essential arity at most $k$} if it has at most $k$ essential coordinates.
Let $(\A,\B)$ be a template.
An $n$-ary function $f$ is a \emph{polymorphism} of $(\A,\B)$ if $(x_i,y_i,z_i)\in
R^\A$ for every $i\in [n]$ implies
\[
  (f(x_1,\ldots,x_n),f(y_1,\ldots,y_n),f(z_1,\ldots,z_n))\in R^\B.
\]
We denote by $\Pol^{(n)}(\A,\B)$ the set of $n$-ary polymorphisms of $(\A,\B)$
and by $\Pol(\A,\B)$ the set of all polymorphisms of $(\A,\B)$.
Polymorphisms form a minion, which we define below.
Given an $n$-ary function $f:A^n\to B$
and a map $\pi : [n] \to [m]$, an $m$-ary function $g:A^m\to B$ is called a \emph{minor} of $f$ given by the map $\pi$ if 
\[
  g(x_1, \ldots, x_m) = f(x_{\pi(1)}, \ldots, x_{\pi(n)}).
\]
We write $f\xrightarrow{\pi} g$ or $g = f^\pi$ if $g$ is the minor of $f$ given by the map $\pi$.
A \emph{minion} on a pair of sets $(A,B)$ is a non-empty set of functions from $A^n$ to $B$ (for $n \in \mathbb{N}$)
that is closed under taking minors.
For any template $(\A,\B)$, the set of polymorphisms $\Pol(\A,\B)$ is
a minion~\cite{BBKO21}.
A map $\psi:\minion{M}\to\minion{N}$ from a minion $\minion{M}$ to a minion $\minion{N}$ is a \emph{minion homomorphism}
if $\psi$ preserves arities, i.e., $\ar(g)=\ar(\psi(g))$ for any $g\in\minion{M}$, 
and $\psi$ preserves minors, i.e., for each $\pi:[n]\to[m]$ and each $n$-ary $g\in\minion{M}$ we have
$\psi(g)(x_{\pi(1)},\ldots,x_{\pi(n)}) = \psi(g(x_{\pi(1)},\ldots,x_{\pi(n)}))$.

\paragraph*{$\bm{\NP}$-hardness}
Certain properties of the polymorphisms of $(\A,\B)$ guarantee the \NP-hardness of
$\PCSP(\A,\B)$~\cite{BBKO21,BWZ21,Barto22:soda}. We use the following
result from~\cite{BWZ21}; see also~\cite{Barto22:soda}.
\begin{theorem}[{\cite[Corollary~4.2]{BWZ21}}]
\label{thm:hardnessCondition}
Let $k, \ell$ be positive integers, and let $(\A,\B)$ be a template. Suppose that $I$ is an assignment that takes any polymorphism $f \in \Pol(\A, \B)$ to a subset of $[\ar(f)]$ of size at most $k$. Suppose that for any chain of minors
\[
f_1 \xrightarrow{\pi_{1,2}} f_2 \xrightarrow{\pi_{2,3}} \cdots \xrightarrow{\pi_{\ell - 1, \ell}} f_\ell
\]
where $f_1, \ldots, f_\ell \in \Pol(\A, \B)$, there exist $1 \leq i < j \leq \ell$ such that $\pi_{i, j}( I( f_i )) \cap I(f_j) \neq \emptyset$, where $\pi_{i, j} = \pi_{j-1, j} \circ \cdots \circ \pi_{i, i+1}$. Then $\PCSP(\A, \B)$ is \NP-hard.
\end{theorem}
Intuitively, one thinks of the mapping $I$ as nondeterministically ``decoding'' a polymorphism $f$ to one of its coordinates. We bound the number of elements in $I(f)$ in order to bound the nondeterminism. For the decoding to be good enough to work, it must satisfy the chain condition in the theorem.

Another important hardness result we will use is the following, which is a direct corollary of~\cite[Corollary 5.2]{BBKO21}, rephrased in terms of essential arity.
\begin{theorem}\label{thm:hardnessCondition2} 
Let $(\A,\B)$ be a template.
If $\Pol^{(3)}(\A, \B)$ contains only non-constant functions of essential arity 1, then $\PCSP(\A, \B)$ is \NP-hard.
\end{theorem}
(In fact, if the condition of Theorem~\ref{thm:hardnessCondition2} holds, then the condition of Theorem~\ref{thm:hardnessCondition} also holds, for $k = 1$ and $\ell = 2$.)
Theorems~\ref{thm:hardnessCondition} and~\ref{thm:hardnessCondition2} hold for both decision and search versions of PCSP.

\paragraph*{Digraphs} 
Unless said otherwise, all digraphs in this paper are finite.
A loopless digraph $\G$ is a \emph{tournament} if, for any two distinct vertices $x,y$, precisely one of the pairs $(x,y)$ and $(y,x)$ is a directed edge of $\G$.
A tournament is \emph{transitive} if its edge relation is transitive. Note that a tournament is transitive if and only if it has no directed cycles of length 3. We will use the well-known result that
a tournament is acyclic if and only if it is transitive, cf.~\cite[Corollary~5a,~(1--2)]{tournaments}. We let $\T_i$ denote the transitive tournament on $i$ vertices.
Following~\cite{notation_old}, an \emph{oriented path} or \emph{oriented cycle} is a digraph formed by choosing an orientation for each edge of a path or cycle.
The \emph{net length} of an oriented path or cycle is the absolute value of the number of forward edges minus the number of backward edges,
for an arbitrary direction of the path or cycle. (By taking the absolute value of this difference, the direction we traverse the path or cycle does not matter.) In contrast, a \emph{directed cycle} is a digraph isomorphic to the digraph with edges $1 \to 2 \to \cdots \to k \to 1$ for some $k \in \mathbb{N}$.
An oriented path is \emph{minimal} if no subpath has a strictly greater net length.
We shall make use of the following result.

\begin{theorem}[{\cite[Claim 1]{notation_old}}] \label{thm:pathjoining}
  Let $\P,\P'$ be minimal oriented paths of the same net length. 
  Then there exists an oriented path $\Q$ that can be homomorphically mapped to $\P$ and $\P'$ with beginnings and ends preserved~%
  (with suitably chosen traversing directions). 
\end{theorem}
\noindent Note that $\Q$ needs to be a minimal path of the same net length as $\P$ and $\P'$. Hence, Theorem~\ref{thm:pathjoining} can be extended to a finite number of minimal oriented paths $\P,\P',\P'',\dots$ of equal net length.

For an integer $i\geq 2$, we let $\D_i$ be the directed cycle on $i$ vertices and $\L_i$ be the directed path on $i$ vertices. We also let $\D_1$ be a single vertex with a loop and $\L_1$ be an isolated vertex, while $\L_\omega=(\mathbb{N}; \{(x,x+1)\mid x\in\mathbb{N}\})$ denotes the infinite directed path.
A digraph $\G$ is \emph{balanced} if each of its
oriented cycles has zero net length.
Equivalently, this means that $\G\to\L_\omega$ \cite{HN04}. We will show
(cf.~Proposition~\ref{prop_basic_balanced}) that this is also equivalent to the condition $\HG\to\zaff$, where $\zaff$ is the relational structure over the domain $\mathbb{Z}$ whose unique, ternary relation is $\{(x,y,z)\in\mathbb{Z}^3\mid x+y+z=1\}$, see Footnote~\ref{footnote_barto_only_holds_if_finite}.

\paragraph*{Trees} We will need the following fact about binary trees that can be easily shown by induction. All of our trees will be rooted.

\begin{lemma}\label{lem:binary-trees}
    A binary tree with more than $2^n$ leaves must have a path from the root to a leaf containing at least $n + 2$ vertices.
\end{lemma}

\paragraph*{Primitive positive formulas, definitions} We will define the notions of primitive positive (pp) formulas and definitions for the special case of relational structures with one ternary relation symbol $R$. A \emph{pp-formula} is an existentially quantified conjunction of (positive) atomic formulas, that are of the form $x = y$ or $R(x, y, z)$ for variables $x, y, z$. For example,
\[
  \exists y \exists z\ R(x, y, z) \land (x = y)
\]
is a pp-formula. For a relational structure $\A = (A; R^\A)$ and a pp-formula $\phi(x, y, z)$ in which only $x, y, z$ are free, we define $\phi^\A$ as a ternary relation that contains a tuple $(a, b, c)$ if and only if 
substituting $(a, b, c)$ for $(x, y, z)$, interpreting existential quantification as being over $A$, and interpreting $R$ as $R^\A$ leads to a true statement. We say that $\phi$ interpreted in $\A = (A; R^A)$ pp-defines $\B = (B; R^B)$ if $A = B$ and $R^B = \phi^\A$.

For templates $(\A, \B)$ and $(\A',\B')$, we say that $(\A, \B)$
\emph{pp-defines} $(\A', \B')$ whenever there exists a pp-formula $\phi(x, y,
z)$ in which precisely $x, y, z$ are free, for which $\phi$ interpreted in $\A$
pp-defines $\A'$, and $\phi$ interpreted in $\B$ pp-defines $\B'$. We shall use
the following result from~\cite{BG21:sicomp}, cf.
also~\cite{Jeavons98:algebraic,Chen09:survey} for
the analogous result for (non-promise) CSPs.

\begin{theorem}
[\cite{BG21:sicomp}]
\label{thm_pp_defn_and_reductions}
    Suppose the template $(\A,\B)$ pp-defines $(\A', \B')$. Then $\PCSP(\A', \B')$ reduces to $\PCSP(\A, \B)$ in logarithmic space.
\end{theorem}

\section{Breaking the promise from the right}
\label{sec_first_theorem}

Let $\G$ be a digraph such that $(\inn{1}{3},\HG)$ is a valid template (equivalently, the edge set of $\G$ is nonempty).
In this section, we will prove the following result.

\newcommand{\toeq}{\ensuremath{\Rightarrow}}

\mainone*

The tractability part of Theorem~\ref{thm:main1} follows from the next, graph-theoretic result.\footnote{The second part of the proposition that follows may be derived from a result in~\cite{Brandts22:phd}. We include a shorter, self-contained proof for completeness.} 

\begin{restatable}{proposition}{propbalanced}
\label{prop_basic_balanced}
For any digraph $\G$, the following holds
\begin{enumerate}
\item $\HG\to \zaff$ if and only if 
$\G\to \L_\omega$ if and only if  $\G$ is balanced, and
\item 
$\zaff \to\HG$ if and only if $\D_i\to \G$ for some $i\in\{1,2,3\}$.
Moreover, if a homomorphism from $\zaff$ to $\HG$ exists, then it can be efficiently computed.
\end{enumerate}
\end{restatable}
\begin{proof}
    For a relational structure $\A = (A; R)$ with one ternary symmetric relation $R$, define $\widetilde{\A}$ to be the graph $(A, E)$ where $(x, y) \in E$ if and only if $(x, x, y) \in R$. Note that $\HG \to \A$ if and only if $\G \to \widetilde{\A}$.\footnote{In other words, $\widehat{\cdot}$ and $\widetilde{\cdot}$ form a \emph{Galois connection} under the preordering given by $\to$.}
    
    Observe that $\widetilde{\zaff}=(\mathbb{Z};\{(a,1-2a):a\in\mathbb{Z}\})$. This structure is a disjoint union of (countably many) forward-infinite directed paths: {First, every vertex has exactly one outgoing and at most one incoming edge~%
    (odd $a$ has in-degree $1$ while even $a$ has in-degree $0$).
    Second, note that $(0,1)$ and $(1,-1)$ are edges in $\widetilde{\zaff}$ and for the remaining edges $(a,1-2a)$ we have $|a|<|1-2a|$ which proves that forward paths lead away from $0$ and that $\widetilde{\zaff}$ is acyclic.} Thus $\widetilde{\zaff} \leftrightarrows \L_\omega$. It follows that $\G \to \widetilde{\zaff}$ if and only if $\G \to \L_\omega$; respectively, $\widetilde{\zaff} \to \G$ if and only if $\L_\omega \to \G$.
    \begin{enumerate}[wide]
        \item $\HG \to \zaff$ if and only if $\G \to \widetilde{\zaff}$, if and only if $\G \to \L_\omega$. For the second equivalence: as $\G$ is finite, $\G \to \L_\omega$ if and only if $\G \to \L_{|\G|}$, which is equivalent to $\G$ being balanced by e.g.~\cite[Proposition~1.13]{HN04}.
    \item For the ``if'' direction: a homomorphism $h_i : \Z \to \HD_i$ exists for $i \in [3]$ (namely $h_1(x) = 0, h_2(x) = [x > 0], h_3(x) = x \bmod 3$), so $\D_i \to \G$ for $i \in [3]$ implies $\zaff \to \HD_i \to \HG$.
        Conversely, let $f:\zaff\to\HG$
        be a homomorphism, and assume that 
        $f$ is surjective~%
        (we can take the image of the map in place of $\HG$) and that $\D_1$ and $\D_2$ do not map to $\G$. 
        We will show that $\G$ is a non-transitive tournament; as such, the graph must contain $\D_3$ and the proof is finished.\footnote{To see why any non-transitive tournament $\mathbf{T}$ must contain $\D_3$, consider the smallest cycle within $\mathbf{T}$. It cannot be of length 1 or 2; if it is of length 3 then we are done; if it is of length 4 or above, then it cannot be minimal since any chord within the cycle gives rise to a smaller cycle.}
        Take any $v\neq v'$ vertices of $\G$ and find $a,a'$ such that
        $f(a)=v,f(a')=v'$.
        The triplet $(a,a',1-a-a')$ belongs to the relation of $\zaff$. 
        Therefore, either $(v,v')$ or $(v',v)$ is an edge of $\G$ --- since $\G$ additionally does not contain directed cycles of length at most 2 by assumption, we have that $\G$ is a tournament.
        Since $f : \Z \to \HG$, we have $\widetilde{\zaff} \to \G$ and hence $\L_\omega \to \G$, so by the pigeonhole principle
        $\G$ must contain a directed cycle and, thus, it
        cannot be transitive.
    \end{enumerate}
    
Finally, if $\Z\to\HG$, using the second part of the lemma we deduce that $\D_i\to\G$ (and hence $\HD_i\to\HG$) for some $i\in\{1,2,3\}$. Both homomorphisms $\HD_i\to\HG$ and $\Z\to\HD_i$ are efficiently computable, and their composition yields a concrete homomorphism from $\Z$ to $\HG$.
\end{proof}

\begin{proof}[Proof of tractability in Theorem~\ref{thm:main1}]
If $\D_i\to\G$ for some $i\in\{1,2,3\}$, it follows from Proposition~\ref{prop_basic_balanced} that $\inn{1}{3}\to \Z \to \HG$. Therefore, $\PCSP(\inn{1}{3}, \HG)$ trivially reduces to $\CSP(\Z)$, which is tractable as it corresponds to solving linear Diophantine systems~\cite{KB79}.
\end{proof}

In the rest of this section, we prove
the hardness part of Theorem~\ref{thm:main1}. We note that this proof can be seen as a generalisation of the proof of hardness of $\PCSP(\inn{1}{3}, \HT_3)$ from~\cite{Barto21:stacs}, where $\T_k$ is a directed tournament on $k$ vertices.\footnote{$\HT_3$ is denoted by $\T_1$ in~\cite{Barto21:stacs}.}
To that end, fix $\G$ and suppose it contains no directed cycles of length at most 3.
If we consider the edge relation of $\G$, we see that it is irreflexive (since $\G$ has no loops) and antisymmetric (since $\G$ has no directed cycles of length 2). Unfortunately it is not transitive --- nonetheless we will see the edge relation as a kind of weak order relation, and we will prove that it is ``transitive enough'' for our purposes. Thus, write $<$ for the edge relation of $\G$ (keeping in mind that $<$ is not in general transitive), and define $\leq, >, \geq$ in the obvious way. Write $x \conn y$ if $x = y$ or $x < y$ or $x > y$.

In this section, all polymorphisms will be from 
the polymorphism minion 
$\Pol(\inn{1}{3}, \HG)$. Note that such polymorphisms are simply functions from $\{0, 1\}^n$ to $V$, the vertex set of $\G$. We can see such a function as a function from $2^{[n]}$ to $V$, i.e.~from subsets of $[n]$ to $V$. The fact that $f \in \Pol^{(n)}(\inn{1}{3}, \HG)$ is a polymorphism then implies that, for any partition $A, B, C$ of $[n]$, two of $f(A), f(B), f(C)$ are equal, and the last is strictly greater (i.e., there is an edge in $\G$ from the two equal elements to the third one). In this interpretation, for $\pi : [n] \to [m]$, observe that $f^\pi = f \circ \pi^{-1}$, where $\pi^{-1} : 2^{[m]} \to 2^{[n]}$ is the preimage function: $x \in \pi^{-1}(S)$ if and only if $\pi(x) \in S$.

\begin{definition}
\label{defn_good_set}
For a polymorphism $f \in \Pol^{(n)}(\inn{1}{3}, \HG)$, call a set $X \subseteq [n]$ a \emph{hitting set} (for $f$) if it has a subset $Y \subseteq X$ such that, for all $Z\subseteq [n]$ with $f(Z) = f(Y)$, we have $X \cap Z \neq \emptyset$. In this case, say that $f(Y)$ is a \emph{hitting value} for $X$.
\end{definition}

Note that the property of being a hitting set is upwards closed; i.e., if $X$ is hitting and $W \supseteq X$, then $W$ is hitting as well. 
\begin{theorem}\label{thm:goodsets}
Let $\G$ be a digraph without directed cycles of length at most 3 and let $N$ be the number of vertices of $\G$. Then every polymorphism $f \in \Pol(\inn{1}{3}, \HG)$ has a nonempty hitting set of size at most $(N+1)(1 + N + N^2 2^{N})$.
\end{theorem}
As we shall see next, the result above is sufficient for establishing the \NP-hardness part of Theorem~\ref{thm:main1}.

\begin{proof}[Proof of Theorem~\ref{thm:main1}]
  For any $f \in \Pol(\inn{1}{3}, \HG)$, define $I(f)$ to be a 
hitting set for $f$ of minimum size --- it exists and its size is bounded by some function of $N$ by Theorem~\ref{thm:goodsets}. Observe that, as $\G$ is fixed, $N$ is a constant, so the size of $I(f)$ is bounded by a constant. Moreover, it is nonempty as $\emptyset$ is never a hitting set.
Consider any $N + 1$ polymorphisms of $\Pol(\inn{1}{3}, \HG)$ connected by a chain of minors. By the pigeonhole principle, polymorphisms $f \xrightarrow{\pi} g$ must exist within this chain such that $I(f)$ and $I(g)$ have a hitting value in common. Suppose this value is $c$. Thus $I(g)$ contains a set $X$ such that $g(X) = c$, and $I(f)$ intersects all sets $Y$ within the domain of $f$ for which $f(Y) = c$. Since $g(X) = c$, we have $f(\pi^{-1}(X)) = c$ and thus $I(f) \cap \pi^{-1}(X) \neq \emptyset$. Thus $\pi(I(f)) \cap I(g) \neq \emptyset$. By Theorem~\ref{thm:hardnessCondition} (with $\ell=N+1$ and $k=(N+1)(1 + N + N^2 2^{N})$), it follows that $\PCSP(\inn{1}{3}, \HG)$ is \NP-hard.
\end{proof}

In the remainder of this section, we shall prove Theorem~\ref{thm:goodsets}. Henceforth, we fix
a directed graph $\G$ with $N$ vertices and no directed cycles of length at most $3$, and
a polymorphism $f \in \Pol^{(n)}(\inn{1}{3}, \HG)$.
We first show that, in this case, $[n]$ is a hitting set for $f$.
\begin{lemma}
\label{lem_all_is_good}
    $[n]$ is a hitting set for $f$.
\end{lemma}
\begin{proof}
   Since $f$ is a polymorphism, the tuple
   \[(f(\emptyset),f(\emptyset),f([n]))\] 
   must belong to the relation of $\HG$, which means that $f(\emptyset)<f([n])$. Since $\G$ has no loops, $f(\emptyset)\neq f([n])$. We then have from Definition~\ref{defn_good_set} that $[n]$ is a hitting set, as witnessed by the subset $[n]\subseteq [n]$.
\end{proof}

Next, we give two laws governing $f$.

\begin{lemma}[Disjointness law]
For disjoint $A, B \subseteq [n]$, we have $f(A) \conn f(B)$.
\end{lemma}
\begin{proof}
Consider the partition $A, B, [n] \setminus (A \cup B)$.
Since $f$ is a polymorphism, the tuple $$(f(A),f(B),f([n] \setminus (A \cup B)))$$ belongs to the relation of $\HG$. Hence, either $f(A)=f(B)<f([n] \setminus (A \cup B))$, or $f(A)=f([n] \setminus (A \cup B))<f(B)$, or $f(B)=f([n] \setminus (A \cup B))<f(A)$.
\end{proof}

\begin{lemma}[Union law]
Consider disjoint $A, B, C \subseteq [n]$. Let $M$ be the multiset $[f(A), f(B), f(A \cup C), f(B \cup C)]$. There exists an element $m$ of $M$ with multiplicity 2 or 4, which we call the \emph{pseudo-minimum} of $M$. If $m$ has multiplicity 2 and $\hat{m}, \tilde{m}$ are the remaining (possibly equal) elements of $M$, then at least one of the following occurs:
\begin{itemize}
    \item $m < \hat{m}$ and $m < \tilde{m}$, or
    \item $m < \hat{m} < \tilde{m}$, or
    \item $m < \tilde{m} < \hat{m}$.
\end{itemize}
\end{lemma}
\begin{proof}
Let $D = [n] \setminus (A \cup B \cup C)$. By disjointness, $f(A \cup C) \conn f(B)$ and $f(A) \conn f(B \cup C)$. If neither of these pairs have equal values, then one element in each pair is equal to $f(D)$ (due to the partitions $A \cup C, B, D$ and $A, B \cup C, D$), and is less than the other value in the pair. This corresponds to the case where $m < \hat{m}$ and $m < \tilde{m}$. Now assume that the elements in at least one of the pairs coincide --- say $f(A) = f(B \cup C)$. If the elements of the other pair are not equal, then one of the two values is equal to $f(D)$, and the other is greater; since $f(A) < f(D)$ (due to partition $A, B \cup C, D$), this corresponds to the case where $m < \tilde{m} < \hat{m}$ or the case where $m < \hat{m} < \tilde{m}$. If also the other pair has equal values, then by disjointness $f(A) \conn f(B)$. If $f(A) \neq f(B)$ we get the case where $m < \hat{m}$ and $m < \tilde{m}$; and if $f(A) = f(B)$ then all four elements of $M$ coincide, so we are in the case where $m$ has multiplicity 4.
\end{proof}
\begin{remark}
The pseudo-minimum is not necessarily a true minimum, since in the case where $m < \hat{m} < \tilde{m}$ (or symmetrically $m < \tilde{m} < \hat{m}$) it is possible that $m \not \leq \tilde{m}$ (respectively $m \not \leq \hat{m}$). Nonetheless, since $\G$ has no directed cycles of length 3, we have that $m \not > \tilde{m}$ (respectively, $m \not > \hat{m}$) even in this case. Indeed, in all cases, for any $m' \in M$ we have that $m \not > m'$, since otherwise a directed cycle of length at most 3 would appear.
\end{remark}

Our final goal is to establish that $f$ has a hitting set of bounded size. The following lemma shows that a long sequence of sets whose images under $f$ are strictly increasing yields a hitting set. 

\begin{lemma}\label{lem:longpaths}
Suppose there exist sets $X_1, \ldots, X_k \subseteq [n]$ with $k>N$, where $f(X_1) < \cdots < f(X_k)$. Then $\bigcup_\ell X_\ell$ is a hitting set.
\end{lemma}
\begin{proof}
Suppose for contradiction that for all $i, j \in [k]$ we have $f(X_i) \conn f(X_j)$.
Thus, $f(X_1), \ldots, f(X_k)$ induce a tournament in $\G$. This tournament is not acyclic, since $f(X_1) < \cdots < f(X_k)$ must contain a directed cycle by the pigeonhole principle. Therefore, $\G$ is not transitive, which means that it must contain a directed cycle of length at most 3, a contradiction.

Thus, there exist $i, j \in [k]$ such that $f(X_i) \not \conn f(X_j)$.
By (the contrapositive of) the disjointness law, any set $Y \subseteq [n]$ such that $f(Y) = f(X_i) \not \conn f(X_j)$
is not disjoint from $X_j$.
Hence, $X_i \cup X_j$ is a hitting set.
Since hitting sets are upwards closed and $\bigcup_\ell X_\ell \supseteq X_i \cup X_j$, we find that $\bigcup_\ell X_\ell$ is a hitting set, too.
\end{proof}

The next three lemmata will partially determine $f(X \cup S \cup T)$ in terms of $f(X \cup S)$ and $f(X \cup T)$, under the assumption that $f(X) = f(Y)$ and $X, Y, S, T$ are disjoint.
These proofs are based on repeated applications of the union law, and involve a case analysis.

\begin{restatable}{lemma}{smallsetinvariance}
\label{lem:smallsetinvariance}
For disjoint $X, Y, S$, if $f(X \cup S) \not > f(X) = f(Y)$ then $f(X \cup S) = f(Y \cup S)$.
\end{restatable}

\begin{proof}
By the disjointness law, $f(X \cup S) \conn f(Y) = f(X)$. Since $f(X \cup S) \not > f(X)$, we have $f(X \cup S) \leq f(X)$.
Apply the union law to $X$, $Y$, and $S$.

Suppose first that the pseudo-minimum of the multiset 
  \[[f(X),f(Y),f(X\cup S),f(Y\cup S)]\]
is $f(X) = f(Y)$. Recall that $f(X \cup S) \leq f(X)$, and since $f(X)$ is the pseudo-minimum, $f(X \cup S) \not < f(X)$. So $f(X \cup S) = f(X)$. Since the pseudo-minimum must appear 2 or 4 times, it follows that $f(Y \cup S) = f(X) = f(Y) = f(X \cup S)$, as required.

Otherwise, if the pseudo-minimum is $f(X \cup S)$ (or symmetrically $f(Y \cup S)$), since the minimum must appear 2 or 4 times, and $f(X) = f(Y)$ is not the pseudo-minimum, we have that $f(X \cup S) = f(Y \cup S)$.
\end{proof}

\begin{restatable}{lemma}{smallsetunion}
\label{lem:smallsetunion}
If $X, Y, S, T$ are disjoint, $f(X \cup S) = f(X \cup T)$, and $f(X) = f(Y)$, then either $f(X) \leq f(X \cup S \cup T)$ or $f(X \cup S) < f(X \cup S \cup T)$.
\end{restatable}

\begin{proof}
    Apply the union law to $X$, $Y$, and $S$. Since $f(X) = f(Y)$, there are two possibilities for the pseudo-minimum of 
    \[[f(X),f(Y),f(X\cup S),f(Y\cup S)].\]
    Namely, if the pseudo-minimum is $f(X\cup S)$ or $f(Y \cup S)$, then since the pseudo-minimum appears 2 or 4 times either all 4 elements are equal or $f(X \cup S) = f(Y \cup S)$ is the pseudo-minimum. Otherwise, the pseudo-minimum must be $f(X) = f(Y)$.
    
    \begin{description}[wide]
    \item[Case 1.] Assume that the pseudo-minimum is $f(X \cup S) = f(Y \cup S)$. Since $f(X) = f(Y)$, we deduce that $f(X \cup S) = f(Y \cup S) \leq f(X) = f(Y)$. Apply the union law to
    $Y$, $X\cup T$, and $S$. Since $f(X \cup T) = f(X \cup S) = f(Y \cup S) \leq f(Y)$, the pseudo-minimum of
    \[
    [f(Y), f(X \cup T), f(Y \cup S), f(X \cup T \cup S)]
    \]
    is $f(X \cup T) = f(Y \cup S)$. 
    We thus have four cases: 
    \begin{itemize}
        \item $f(X \cup T) = f(Y \cup S) = f(Y) = f(X \cup T \cup S)$,
        \item $f(Y) > f(X \cup T) = f(Y \cup S) < f(X \cup T \cup S)$,
        \item $f(X \cup T) = f(Y \cup S) < f(Y) < f(X \cup T \cup S)$,
        \item $f(X \cup T) = f(Y \cup S) < f(X \cup T \cup S) < f(Y)$.
    \end{itemize}
    Keeping in mind that $f(X) = f(Y)$ and $f(X \cup S) = f(Y \cup S)$, the conclusion follows in all 4 cases.

    \item[Case 2.] Assume that the pseudo-minimum is equal to neither $f(X \cup S)$ nor $f(Y \cup S)$. Thus the pseudo-minimum is $f(X) = f(Y)$. In this case, $f(Y) \not > f(X \cup S) = f(X \cup T)$ and $f(Y) \not > f(Y \cup S)$. By assumption, $f(Y) \neq f(X \cup T)$ and $f(Y) \neq f(Y \cup S)$. By disjointness, $f(Y) \conn f(X \cup T)$ and $f(Y) = f(X) \conn f(Y \cup S)$, so it follows that $f(Y) < f(X \cup T)$ and $f(Y) < f(Y \cup S)$.
    Now, apply the union law to $Y$, $X\cup T$, and $S$. Since $f(Y) < f(X \cup T)$ and $f(Y) < f(Y \cup S)$, the pseudo-minimum of
    \[
    [f(Y), f(X \cup T), f(Y \cup S), f(X \cup T \cup S)]
    \]
    must be $f(Y) = f(X \cup T \cup S)$, and the conclusion follows.\qedhere
    \end{description}
\end{proof}

\begin{restatable}{lemma}{equalsetunion}
\label{lem:equalsetunion}
If $X, Y, S, T$ are disjoint and $f(X) = f(Y) = f(X \cup S) = f(X \cup T)$, then $f(X) = f(X \cup S \cup T)$.
\end{restatable}

\begin{proof}
Since $\G$ has no loops, $f(X \cup S) \not > f(X)$ and $f(X \cup T) \not > f(X)$. Therefore, by Lemma~\ref{lem:smallsetinvariance}, $f(Y \cup S) = f(X \cup S) = f(X)$ and $f(Y \cup T) = f(X \cup T) = f(X)$. Let us now apply the union law to $Y$, $X\cup S$, and $T$. Since three of the elements of $[f(Y),f(X\cup S), f(Y\cup T), f(X\cup S\cup T)]$
 are equal, the fourth is also.
\end{proof}

The next lemma is the crucial one: Given some nonempty set $X$, it either creates a hitting set of bounded size immediately, or finds some set $X'$ of bounded size such that $f(X) < f(X')$. The conclusion will follow by repeatedly applying this lemma.

\begin{lemma}\label{lem:stepup}
Consider a nonempty set $X\subseteq [n]$. Then either a hitting set of size at most $|X| + 2^{N}$ exists, or we can find a nonempty set $X'\subseteq [n]$ of size at most $1 + |X| + N 2^{N}$ such that $f(X) < f(X')$.
\end{lemma}
\begin{proof}
We can assume that $|X|\leq n-2$ as, otherwise,
it would follow from Lemma~\ref{lem_all_is_good} that $[n]$ is the required hitting set. Thus, let $a\neq b \in [n] \setminus X$, and suppose that $X \cup \{a, b\}$ is not hitting (if it were, it would be the required hitting set). Thus, some set $Y \subseteq [n] \setminus (X \cup \{a, b\})$ exists such that $f(X) = f(Y)$. Let $S = [n] \setminus (X \cup Y \cup \{a\})$; since $b \in S$, $S \neq \emptyset$.

Create a partition $P_1, \ldots, P_k$ of $S$, that initially consists of one singleton for each element of $S$. While
\begin{enumerate}
    \item \label{cond1} all the parts $P$ such that $f(X \cup P) < f(X)$ have size $|P| \leq 2^{N-1}$, and
    \item \label{cond2} there exist distinct parts $P_i, P_j$ such that $f(X \cup P_i) = f(X \cup P_j) < f(X)$,
\end{enumerate}
merge any two such parts $P_i$ and $P_j$. Observe that at all times all parts must have size at most $2^{N}$, since they are either singletons or the union of two parts with size at most $2^{N - 1}$. It follows that this procedure must eventually terminate; we now consider what happens when it does, depending on the reason for termination.

First, suppose that the procedure terminates because \eqref{cond1} ceases to hold; i.e., we arrive at a part $P$ with size greater than $2^{N - 1}$ with $f(X \cup P) < f(X)$. $P$ was created by repeated merges of parts in the partition; thus, consider a binary tree rooted at $P$ where each vertex is labelled by a subset of $P$ that was at some point a part in the partition, and the children of a vertex are the parts that were merged to form that part. Consider any non-root vertex in the tree, and suppose it is labelled by part $P'$. Since $P'$ was merged with some other part, it must be the case that $f(X \cup P') < f(X)$. Since this is true by assumption for $P$ as well, it is true for all the parts that appear in the tree. Now, consider any non-leaf vertex, labelled by part $Q \cup R$, where its children are labelled by $Q$ and $R$. Since $Q$ and $R$ are merged, $f(X \cup Q) = f(X \cup R)$ and $Q \cap R = \emptyset$. Thus, apply Lemma~\ref{lem:smallsetunion} to $X, Y, Q, R$
and note that $f(X) \not \leq f(X \cup Q \cup R)$, since $f(X \cup Q \cup R) < f(X)$; it follows that $f(X \cup Q) < f(X \cup Q \cup R)$ and $f(X \cup R) < f(X \cup Q \cup R)$. In other words, for any two parts $A$ and $B$ where the vertex corresponding to $A$ is a child of the vertex corresponding to  $B$, we have $f(
X\cup A) < f(X\cup B)$.
Since the tree has more than $2^{N - 1}$ leaves, by Lemma~\ref{lem:binary-trees} we can find a path in the tree starting at the root with at least $N + 1$ vertices.
Call the labels of the vertices of such path $P_k'=P, P_{k-1}',\dots,P_1'$, starting from the root and going to the leaves.
Since $f(X\cup P_1')<f(X\cup P_2')<\dots<f(X\cup P_k')$, we obtain from Lemma~\ref{lem:longpaths} that
the set $X\cup\bigcup_{\ell\in [k]}P_\ell'$ is a hitting set.
Using that $P$ is a superset of all sets appearing in the path and that hitting sets are upwards closed, we conclude that $X\cup P$ is a hitting set.
Recall that $|P| \leq 2^{N}$. So, the required hitting set is $X\cup P$.

Second, suppose that the procedure terminates because the condition \eqref{cond2} ceases to hold (regardless of whether \eqref{cond1} holds or not). By the pigeonhole principle there exist at most $N$ parts $P$ for which $f(X \cup P) < f(X)$. Let $Z$ be the union of these parts and the set $\{a\}$ (if there are not any such parts then $Z = \{a\}$), and let $Q_1, \ldots, Q_\ell$ be the remaining parts; i.e., $f(X \cup Q_i) \not < f(X)$. Observe that $|Z| \leq 1 + N 2^{N}$. There are now two cases. First, suppose that for some $i\in [\ell]$ we have $f(X) < f(X \cup Q_i)$. In this case, since $|X \cup Q_i| \leq |X| + 2^{N}$, the set $X' = X \cup Q_i$ witnesses that the statement of the lemma holds. Otherwise, for every $i \in [\ell]$, since $f(X) \not < f(X \cup Q_i)$ and $f(X \cup Q_i) \not < f(X)$, yet $f(X) = f(Y) \conn f(X \cup Q_i)$ by disjointness, we find that $f(X) = f(X \cup Q_i)$. By Lemma~\ref{lem:equalsetunion} applied $\ell-1$ times, $f(X \cup \bigcup_i Q_i) = f(X) = f(Y)$. (This also holds when there are no sets $Q_i$, since $X \cup \bigcup_i Q_i = X$ in this case.) Now, the partition $X \cup \bigcup_i Q_i, Y, Z$ implies that $f(X \cup \bigcup_i Q_i) = f(Y) < f(Z)$ (as $f$ is a polymorphism), and thus $f(X) = f(Y) < f(Z)$. Thus, as $|Z| \leq 1 + N2^{N}$, the set $X' = Z$ witnesses that the statement of the lemma holds.
\end{proof}

\begin{proof}[Proof of Theorem~\ref{thm:goodsets}]
  Create a sequence of nonempty sets \[X_1, \ldots, X_{N+1} \subseteq [n]\] in the following way. Let $X_1 = \{1\}$. For $i\in [N]$, apply Lemma~\ref{lem:stepup} to $X_i$. If it yields a nonempty hitting set of size at most $|X_i| + 2^{N}$, then let $X_{i+1}$ be this hitting set. Otherwise, if it yields a nonempty set $X'$ of size at most $1 + |X_i| + N2^{N}$ for which $f(X_i) < f(X')$, set $X_{i+1} = X'$. Note that $|X_{i+1}| \leq 1 + |X_{i}| + N2^{N}$ and $|X_{1}| = 1$, so, for all $i \in [N]$, $|X_{i+1}| \leq 1 + i (1 + N2^{N}) \leq 1 + N(1 + N2^N) = 1 + N + N^2 2^{N}$. If $X_{\ell}$ is a hitting set for some $\ell\in [N+1]$, the conclusion follows. Otherwise, we must have $f(X_1) < f(X_2) < \cdots < f(X_{N+1})$, in which case, by Lemma~\ref{lem:longpaths}, $\bigcup_i X_i$ is a hitting set. Since $|\bigcup_i X_i| \leq (N+1)(1 + N + N^2 2^{N})$, the conclusion follows in this case, too.
\end{proof}

\section{Breaking the promise from the left}
\label{sec_second_theorem}

Let $\G$ be a digraph such that $(\HG,\NAE)$ is a valid template (equivalently, the graph obtained from $\G$ by forgetting the directions is bipartite; in this case, we say that $\G$ is bipartite). In this section, we will prove the following result.

\maintwo*

The tractability part is a direct application of Proposition~\ref{prop_basic_balanced}.

\begin{proof}[Proof of tractability in Theorem~\ref{thm:main2}]
If $\G$ is balanced, applying both parts of Proposition~\ref{prop_basic_balanced}, we find
\begin{align*}
    \HG\to\Z\to\HD_2=\NAE.
\end{align*}
Therefore, $\PCSP(\HG,\NAE)$ reduces to $\CSP(\Z)$ and is thus tractable.
\end{proof}

We now turn to prove the hardness part of Theorem~\ref{thm:main2}.
Suppose that $\G$ is unbalanced --- i.e., $\G \not \to \L_\omega$. 
Note that, since $\G$ is bipartite, the net length of any oriented cycle is even.
The proof of hardness in Theorem~\ref{thm:main2} shall follow from the combination of the next two facts.

\begin{restatable}{proposition}{cycleHard}
\label{thm:cycleHard}
For any positive integer $k$, $\PCSP(\HD_{2k}, \NAE)$ is \NP-hard.
\end{restatable}
\begin{restatable}{proposition}{redCycle}\label{thm:redCycle}
For any bipartite digraph $\G$ that contains an oriented cycle with net length
  $2k \neq 0$, $\PCSP(\HD_{2k}, \NAE)$ reduces in polynomial time to $\PCSP(\HG, \NAE)$.
\end{restatable}

These two results will be proved in Section~\ref{sec:cycleHard} and
Section~\ref{sec:redCycle}, respectively. 

\begin{proof}[Proof of hardness in Theorem~\ref{thm:main2}]
Suppose that $\G$ contains
an oriented cycle of nonzero net length. Since $\G$ is bipartite, this net length must be even, say $2k$.
By Proposition~\ref{thm:redCycle},
$\PCSP(\HD_{2k}, \NAE)$ reduces to $\PCSP(\HG, \NAE)$. \NP-hardness of $\PCSP(\HG, \NAE)$ then follows by
Proposition~\ref{thm:cycleHard}.
\end{proof}

\subsection{Hardness of cycle vs.~NAE}\label{sec:cycleHard}
In this section, we prove the following result.

\cycleHard*

In the following proof, we make use of the fact that, for 
$f \in \Pol^{(3)}(\HD_{2k}, \NAE)$,
\begin{equation}\label{fact}
\left.
\begin{gathered}
f(x, y, z) = f(x', y', z') \\
(x, x', x''), (y, y', y''), (z, z', z'')  \in R^{\HD_{2k}}
\end{gathered}\right\} \quad \Rightarrow \quad f(x, y, z) \neq f(x'', y'', z'')
\end{equation}
This fact follows directly from the definitions of polymorphisms and of $\NAE$. Furthermore,
\begin{equation}\label{factBool}
f(x, y, z) \neq f(x', y', z') \neq f(x'', y'', z'') \quad \Rightarrow \quad f(x, y, z) = f(x'', y'', z''),
\end{equation}
since $\NAE$ is Boolean. Finally, since $\HD_{2k}$ can be described as the template whose domain is $[0, 2k)$
and whose relation contains all permutations of tuples of the form $(x, x, x+1 \bmod 2k)$, we will consider addition over $[0, 2k)$ to be done modulo $2k$. In particular, these facts imply
\begin{equation}\label{fact2}
    f(x, y, z) \neq f(x+1, y+1, z+1),
\end{equation}
for $x, y, z \in [0, 2k)$.

\newcommand{\wildcard}{\ensuremath{\star}}

\begin{proof}
The result clearly holds for $k = 1$, since in this case $\PCSP(\HD_{2k}, \NAE) = \CSP(\NAE)$, which is \NP-hard. Thus, assume $k \geq 2$. We show that every ternary polymorphism of $(\HD_{2k}, \NAE)$ is non-constant and has essential arity 1.
This is  sufficient for hardness by Theorem~\ref{thm:hardnessCondition2}.
Since, by \eqref{fact2}, $f(x, y, z) \neq f(x+1, y+1, z+1)$ for any $x, y, z \in[0, 2k)$, we deduce by induction that $f(x, y, z) = z + f(x - z, y - z, 0) \bmod 2$ for $x, y, z\in [0, 2k)$. It thus follows that it is sufficient to describe the 
matrix $M_{xy} = f(x, y, 0)$ in order to characterise $f$ entirely. We now observe the following.
\begin{enumerate}
    \item\label{equalLine} 
    If $M$ contains two adjacent equal elements, then the row or column on which they are found is constant.
    To see why, suppose without loss of generality that $f(x,y,0) = M_{xy} = M_{x, y+1} = f(x,y + 1,0)$ for some $x, y\in [0, 2k)$. Thus, by \eqref{fact}, applied to $(x, y, 0), (x, y+1, 0), (x+1, y, 1)$, we have $f(x, y, 0) \neq f(x + 1, y, 1) $. Furthermore, by \eqref{fact2}, $f(x, y-1, 0) \neq f(x+1, y, 1)$. Thus, by \eqref{factBool}, $M_{x, y-1} = f(x, y-1, 0) = f(x, y, 0) = M_{xy}$.
    Repeating this observation $2k-2$ times yields the result. 
    \item\label{distinctLine} If a row or column of $M$ is non-constant, then it alternates between 0 and 1. This is just the contrapositive of \eqref{equalLine}.
    \item\label{lRule} The following configurations cannot appear in $M$:
    \[
    \begin{bmatrix}
    0 & \wildcard &\wildcard  \\
     \wildcard & \wildcard & 1 \\
      \wildcard & 1 & \wildcard \\
    \end{bmatrix},
    \begin{bmatrix}
    1 & \wildcard &\wildcard  \\
     \wildcard & \wildcard & 0 \\
      \wildcard & 0 & \wildcard \\
    \end{bmatrix}.
    \]
     Equivalently, letting $(x, y)$ be the coordinates of the middle element in the patterns above: if $M_{x + 1, y} = M_{x, y + 1}$ then 
     $M_{x-1, y-1} = M_{x + 1, y}$. To see why this is the case, note that $f(x + 1, y, 0) = M_{x + 1, y} = M_{x, y + 1} = f(x, y + 1, 0)$, so by \eqref{fact} applied to $(x + 1, y, 0), (x, y + 1, 0), (x, y, 1)$, we have that  $M_{x+1, y} = f(x+1, y, 0) \neq f(x, y, 1)$. Furthermore, by \eqref{fact2}, $M_{x-1,y-1} = f(x-1, y-1, 0) \neq f(x, y, 1)$. So, by \eqref{factBool}, $M_{x-1, y-1} = M_{x, y + 1}$.
\end{enumerate}
    These facts together show that $M$ is completely determined by its $2 \times 2$ submatrix $M'$ located in the upper-left corner of $M$: These four entries dictate the rows and columns they are on by \eqref{equalLine} and \eqref{distinctLine}, and these rows and columns determine all of $M$ in the same way.
    By propagating out in this way, one easily shows that, if
    \[
    M' \in \left\{
    \begin{bmatrix}
    0 & 0 \\
    0 & 0 \\
    \end{bmatrix},
    \begin{bmatrix}
    1 & 1 \\
    1 & 1 \\
    \end{bmatrix},
    \begin{bmatrix}
    1 & 0 \\
    1 & 0 \\
    \end{bmatrix}, 
    \begin{bmatrix}
    0 & 1 \\
    0 & 1 \\
    \end{bmatrix},
    \begin{bmatrix}
    1 & 1 \\
    0 & 0 \\
    \end{bmatrix},
    \begin{bmatrix}
    0 & 0 \\
    1 & 1 \\
    \end{bmatrix}
    \right\},
    \]
    then $f$ is non-constant and has essential arity 1.\footnote{In particular, $f$ is one of the 6 functions $(x_1, x_2, x_3) \mapsto x_i + c \bmod 2$, for $i \in [3], c \in [2]$.}
    For the remaining ten configurations, by propagating out we see that the $4
    \times 4$ submatrix located in the upper-left corner of $M$ must follow one of
    the five patterns in Figure~\ref{fig:patterns}
or one of the patterns obtained from those in Figure~\ref{fig:patterns} by swapping $0$s and $1$s.
These patterns do not respect rule \eqref{lRule} due to the shaded elements.
\end{proof}

    \begin{figure*}[thb]
      \begin{center}
    \tikz[baseline=(M.west)]{%
    \node[matrix of math nodes,matrix anchor=west,left delimiter={[},right delimiter={]},ampersand replacement=\&] (M) {%
     0 \& 1 \& 0 \& 1 \\
     1\& 1 \& 1 \& 1 \\
     0\& 1 \& 0 \& 1 \\
     1\& 1 \& 1 \& 1 \\
    };
    \begin{pgfonlayer}{bg}
    \node[draw,fit=(M-1-1)(M-2-2),inner sep=-1pt,rounded corners=4pt] {};
    \node[draw,fill=lightgray, inner sep=-1pt, rounded corners=4pt, fit=(M-1-1)(M-1-1), fill opacity=0.5]{};
    \node[draw,fill=lightgray, inner sep=-1pt, rounded corners=4pt, fit=(M-2-3)(M-2-3), fill opacity=0.5]{};
    \node[draw,fill=lightgray, inner sep=-1pt, rounded corners=4pt, fit=(M-3-2)(M-3-2), fill opacity=0.5]{};
    \end{pgfonlayer}
  },
    \tikz[baseline=(M.west)]{%
    \node[matrix of math nodes,matrix anchor=west,left delimiter={[},right delimiter={]},ampersand replacement=\&] (M) {%
     1\& 0 \& 1 \& 0 \\
     1\& 1 \& 1 \& 1 \\
     1\& 0 \& 1 \& 0 \\
     1\& 1 \& 1 \& 1 \\
    };
    \begin{pgfonlayer}{bg}
    \node[draw,fit=(M-1-1)(M-2-2),inner sep=-1pt,rounded corners=4pt] {};
    \node[draw,fill=lightgray, inner sep=-1pt, rounded corners=4pt, fit=(M-1-2)(M-1-2), fill opacity=0.5]{};
    \node[draw,fill=lightgray, inner sep=-1pt, rounded corners=4pt, fit=(M-2-4)(M-2-4), fill opacity=0.5]{};
    \node[draw,fill=lightgray, inner sep=-1pt, rounded corners=4pt, fit=(M-3-3)(M-3-3), fill opacity=0.5]{};
    \end{pgfonlayer}
  },
    \tikz[baseline=(M.west)]{%
    \node[matrix of math nodes,matrix anchor=west,left delimiter={[},right delimiter={]},ampersand replacement=\&] (M) {%
     1\& 1 \& 1 \& 1 \\
     0\& 1 \& 0 \& 1 \\
     1\& 1 \& 1 \& 1 \\
     0\& 1 \& 0 \& 1 \\
    };
    
    \begin{pgfonlayer}{bg}
    \node[draw,fit=(M-1-1)(M-2-2),inner sep=-1pt,rounded corners=4pt] {};
    \node[draw,fill=lightgray, inner sep=-1pt, rounded corners=4pt, fit=(M-2-1)(M-2-1), fill opacity=0.5]{};
    \node[draw,fill=lightgray, inner sep=-1pt, rounded corners=4pt, fit=(M-4-2)(M-4-2), fill opacity=0.5]{};
    \node[draw,fill=lightgray, inner sep=-1pt, rounded corners=4pt, fit=(M-3-3)(M-3-3), fill opacity=0.5]{};
    \end{pgfonlayer}
  },
    \tikz[baseline=(M.west)]{%
    \node[matrix of math nodes,matrix anchor=west,left delimiter={[},right delimiter={]},ampersand replacement=\&] (M) {%
     1\& 1 \& 1 \& 1 \\
     1\& 0 \& 1 \& 0 \\
     1\& 1 \& 1 \& 1 \\
     1\& 0 \& 1 \& 0 \\
    };
    \begin{pgfonlayer}{bg}
    \node[draw,fit=(M-1-1)(M-2-2),inner sep=-1pt,rounded corners=4pt] {};
    \node[draw,fill=lightgray, inner sep=-1pt, rounded corners=4pt, fit=(M-2-2)(M-2-2), fill opacity=0.5]{};
    \node[draw,fill=lightgray, inner sep=-1pt, rounded corners=4pt, fit=(M-4-3)(M-4-3), fill opacity=0.5]{};
    \node[draw,fill=lightgray, inner sep=-1pt, rounded corners=4pt, fit=(M-3-4)(M-3-4), fill opacity=0.5]{};
    \end{pgfonlayer}
  },
    \tikz[baseline=(M.west)]{%
    \node[matrix of math nodes,matrix anchor=west,left delimiter={[},right delimiter={]},ampersand replacement=\&] (M) {%
     0\& 1 \& 0 \& 1 \\
     1\& 0 \& 1 \& 0 \\
     0\& 1 \& 0 \& 1 \\
     1\& 0 \& 1 \& 0 \\
    };
    \begin{pgfonlayer}{bg}
    \node[draw,fit=(M-1-1)(M-2-2),inner sep=-1pt,rounded corners=4pt] {};
    \node[draw,fill=lightgray, inner sep=-1pt, rounded corners=4pt, fit=(M-1-1)(M-1-1), fill opacity=0.5]{};
    \node[draw,fill=lightgray, inner sep=-1pt, rounded corners=4pt, fit=(M-2-3)(M-2-3), fill opacity=0.5]{};
    \node[draw,fill=lightgray, inner sep=-1pt, rounded corners=4pt, fit=(M-3-2)(M-3-2), fill opacity=0.5]{};
    \end{pgfonlayer}
  }
    \end{center}
    \caption{Patterns from the proof of Proposition~\ref{thm:cycleHard}.}
    \label{fig:patterns}
\end{figure*}

\subsection{Reduction}
\label{sec:redCycle}

In this section, we prove the following result.
\redCycle*

\newcommand{\theedge}{\ensuremath{\L_2}}

    In order to design the reduction, we will need a key insight, which we prove now. The reasoning builds up on an idea in~\cite{simple_Xuding}. In the following, for two oriented paths $\P=p_1\dots p_n$ and $\Q=q_1\dots q_m$, we let $\P + \Q$ denote the path formed by taking the disjoint union of $\P$ and $\Q$ and then identifying $p_n$ with $q_1$; and we let $\P - \Q$ denote the path formed by identifying $p_n$ with $q_m$. Furthermore, recall that $\theedge$ denotes the oriented path $p_1p_2$ with one directed edge $(p_1,p_2)$.
    \begin{proposition}\label{prop:easy}
      Let $\C$ be an oriented cycle of net length $n \geq 1$. 
      There exists an oriented path $\P$ and a set of vertices $a_0,\dotsc, a_{n-1}$ of $\C$ such that 
      \begin{enumerate}
        \item for each $i$, $\P - \P$ homomorphically maps into $\C$ in a way that its endpoints map to $a_i$, and
        \item for each $i$, $\theedge+\P-\P$ homomorphically maps into $\C$ in a way that its first vertex maps to $a_i$ and its last vertex maps to $a_{i + 1 \bmod n}$.
      \end{enumerate}
    \end{proposition}
    \begin{proof}
    We first choose a positive orientation of $\C$, i.e., an orientation such that the difference between the numbers of edges directed forward and backward is positive (and thus equal to the net length). 
    
    We now proceed to choose a ``starting point''.
    Label the vertices of $\C$ as $c_0,c_1,c_2,\dots,c_p=c_0$, 
    where $c_0$ is an arbitrary vertex and the other vertices are picked according to the positive orientation, in the natural way. 
    Let $\widetilde{\C}$ be the oriented path obtained from $\C$ by deleting the edge connecting $c_0$ and $c_{p-1}$. 
    Assign levels to vertices in a natural way: $c_0$ has level $0$~%
    (denoted $\lvl(c_0)=0$),
    $c_1$ has level $1$ or $-1$ depending on the direction of the first edge,
    and so on.
    Let $\mu$ be the minimum level of the vertices of $\widetilde{\C}$~%
    (note that this is a non-positive number), 
    and let $v$ be the last vertex of $\widetilde{\C}$ of level $\mu$. 
    
    Let now $\C^*$ be an oriented path starting at $v$ and winding around $\C$ a
    certain amount of times~(the number is inessential, as long as it is sufficiently large) in the positive orientation.
    Assign levels to vertices of the oriented path $\C^*$ in the same way as before.
    Clearly, $\lvl(v)= 0$; the choice of $v$ guarantees that the levels of all other vertices are strictly positive. (Indeed, the first time we wind around $\C$ this is the case, and every successive time we wind around $\C$ the levels are increased by the net length of $\C$, which is positive.)
    Further, for $i=0,\dots,n$, we let $a_i$ be the last vertex in $\C^*$ of level $i$.
    Note that $a_0,\dots,a_{n-1}$ appear in the first copy of $\C$. Moreover, $a_0=v$, while $a_n$ coincides with the duplicate of $v$ in the second copy of $\C$.
    It remains to find the oriented path $\P$. 

    First, for every $i=0,\dots,n-1$, we will find a minimal path $\P_i$ of net length $p$ so that $\P_i-\P_i$ connects $a_i$ to itself.
    This is very easy: We start with $a_i$ and go along $\C^*$ until we hit the first element of level $i+p$.
    This is our $\P_i$;
    the path is minimal since no element after $a_i$ has level $i$,
    and since we chose the first element of level $i+p$ to be the end.
    Moreover, $a_i$ is clearly connected to itself via $\P_i-\P_i$ in $\C$.

    Next, for each $i=0,\dots,n-1$, we let $a'_i$ be the element on $\C^*$ coming after $a_i$.
    Clearly, by the definition of $a_i$, $\lvl(a'_i)=i+1$, and all the vertices of $\C^*$ after $a'_i$ have level $\geq i+1$.
    As before, we go along $\C^*$ until we hit the first vertex of level $i+1+p$, say $b_i$. 
    Let $\Q_i$ be the part of $\C^*$ connecting $a'_i$ to $b_i$;
    it has net length $p$ and it is minimal.
    By the choice of $p$, along the way, we hit the element $a_{i + 1}$.
    Let $\Q_i'$ be the part of $\Q_i$ connecting $a_{i + 1}$ to $b_i$; it is clearly minimal by the choice of $a_{i + 1}$ and $b_i$, and it has net length $p$.
    Additionally, $\theedge + \Q_i - \Q_i'$ connects $a_i$ to $a_{i + 1}$ 
    by construction.

    Now let $\P$ be the oriented path homomorphically mapping to every $\P_i$, $\Q_i$, and $\Q'_i$~%
    (it exists by Theorem~\ref{thm:pathjoining}).
    Using that $a_0$ and $a_n$ are the same vertex in $\C$, it follows that $\P$ satisfies the conclusion of the lemma.
    \end{proof}

    \begin{remark}
    \label{cor:easy}
Proposition~\ref{prop:easy} also holds if one replaces $\C$ with any directed graph $\G$ that contains an oriented cycle of net length $n\geq 1$. Indeed, one only needs to apply the result to any oriented cycle contained in $\G$.
    \end{remark}

  \paragraph*{Reduction}
    We reduce $\PCSP(\HD_{2k}, \NAE)$ to $\PCSP(\HG, \NAE)$. For any oriented path $\P$ with vertex set $\{ v_1, \ldots, v_n \}$, whose first vertex is $v_1$, whose last vertex is $v_n$, and whose edge set is $E(\P)$, define the pp-formula $x \xrightarrow{\P} y$ by
    \[
    \exists u_1, \ldots, u_n\ (x = u_1) \land (y = u_n) \land \bigwedge_{(v_i, v_j) \in E(\P)} R(u_i, u_i, u_j).
    \]
    This pp-formula is such that, in the structure $\HG$, $x \xrightarrow{\P} y$ is true if and only if $\P$ homomorphically maps into $\G$ in a way that $v_1$ is mapped to $x$ and $v_n$ is mapped to $y$.
    \begin{lemma}\label{lem:phi}
      There is a pp-formula $\phi$ such that
      \begin{enumerate}
        \item $\phi$ interpreted in $\HD_i$ pp-defines $\HD_i$ for any $i\in\mathbb{N}$,
        and
        \item if $\phi$ interpreted in $\HG$ pp-defines $\T$, then $\HD_{2k}$ homomorphically maps into $\T$.
      \end{enumerate}
    \end{lemma}
    \begin{proof}
    Consider the path $\P$ for $\G$ provided by Proposition~\ref{prop:easy} (see Remark~\ref{cor:easy}),
    and define the formula $\phi(x,y,z)$ by
    \begin{equation*}
    \exists x', y', z'\ R(x', y', z') \land x \xrightarrow{\P - \P} x' \land y \xrightarrow{\P - \P} y' \land z \xrightarrow{\P - \P} z'.
    \end{equation*}
    This is not a pp-formula, but it is equivalent to one, by ``pulling out'' the existential quantifiers hidden in $\wildcard \xrightarrow{\P - \P} \wildcard$.

    To see why the first item follows, note that $\P-\P$ has net length $0$ and thus, in $\HD_i$, we have $x'=x$, $y'=y$, and $z'=z$.
    
      To see why the second item follows, consider the vertices \[a_0, \ldots,
      a_{2k-1}\] given by Proposition~\ref{prop:easy}. Suppose we label the
      vertices of $\HD_{2k}$ by $0, 1, \ldots, 2k-1$. Then the desired homomorphism is the one sending $i$ to $a_i$. Letting $\oplus$ be addition modulo $2k$, that this is a homomorphism is
      equivalent to saying that 
      \[
          \phi(a_i, a_i, a_{i \oplus 1}) 
          \land \phi(a_i, a_{i \oplus 1}, a_i) 
          \land \phi(a_{i \oplus 1}, a_i, a_i)
      \]
     is true 
      in $\HG$ for all $i=0,\dots,2k-1$. Since $\phi$ is clearly symmetric, it
      is only necessary to prove that $\phi(a_i, a_i, a_{i \oplus 1})$
      holds. Recall that $\theedge + \P - \P$ homomorphically maps into $\G$ such
      that its beginning is mapped to $a_i$ and its end is mapped to $a_{i \oplus 1}$. In other words, there exists some vertex $b$ such that an edge from $a_i$ to $b$ exists in $\G$, and $\P- \P$ maps into $\G$ such that its beginning is mapped to $b$ and its end is mapped to $a_{i \oplus 1}$. As a consequence, $b \xrightarrow{\P - \P} a_{i \oplus 1}$ is true when interpreted in $\HG$; by the symmetry of $\P - \P$, $a_{i \oplus 1} \xrightarrow{\P - \P} b$ is true as well. Thus, the witnesses to the truth of $\phi(a_i, a_i, a_{i \oplus 1})$ are $x' = a_i, y' = a_i$, and $z' = b$. Indeed, $R(a_i, a_i, b)$ is true since there is an edge from $a_i$ to $b$; $a_i \xrightarrow{\P - \P} a_i$ is true by Proposition~\ref{prop:easy}; and $a_{i \oplus 1} \xrightarrow{\P - \P} b$ is true as shown above.
    \end{proof}

\begin{proof}[Proof of Proposition~\ref{thm:redCycle}]
Applying Lemma~\ref{lem:phi}, we find that there exists a pp-formula $\phi$ such that, if $\phi$ interpreted in $\HG$ pp-defines $\T$, then $\HD_{2k} \to \T$, and $\phi$ interpreted in $\NAE=\HD_2$ pp-defines $\NAE$.
Hence, $(\HG,\NAE)$ pp-defines $(\T,\NAE)$. By Theorem~\ref{thm_pp_defn_and_reductions}, $\PCSP(\T,\NAE)$ reduces to $\PCSP(\HG,\NAE)$. On the other hand, we have that $\PCSP(\HD_{2k},\NAE)$ reduces to $\PCSP(\T,\NAE)$, since $\HD_{2k}\to\T$. Combining the two reductions, we obtain the required result.
\end{proof}

\section{Future directions}
\label{sec_conclusion}

The results obtained in the current work identified two features of the satisfiability problem \textsc{1-in-3} vs.~\textsc{NAE} that can be regarded as the reason for its tractability: The problem is solvable in polynomial time because \textsc{1-in-3} corresponds to a balanced digraph, or because $\NAE$ corresponds to a small cycle.
We 
completely classified the
complexity of the extensions of the \textsc{1-in-3} vs.~\textsc{NAE} problem obtained by breaking the promise either from the left or from the right, in the symmetric rainbow-free regime. 
It is noteworthy that the classifications for the two cases are in terms of structural properties of the underlying graphs having two different natures. Indeed, the key property guaranteeing tractability of templates of the form $(\inn{1}{3},\HG)$ is \emph{local}, in that it corresponds to the presence of small directed cycles in $\G$. In contrast, the balancedness property, regulating the complexity of templates of the form $(\HG,\NAE)$, cannot be detected by looking at small subgraphs of $\HG$, and is thus a \emph{global} property.

As mentioned in the Introduction, our results fit within the broader picture of the complexity investigation of PCSPs. 
In this sense, we see two natural directions for future analysis:
\begin{enumerate}
    \item breaking the promise from \emph{both} sides simultaneously;
    \item relaxing the symmetricity and rainbow-freeness assumptions.
\end{enumerate}

The direction (i) corresponds to studying the complexity of problems $\PCSP(\HG, \HH)$ for arbitrary pairs of digraphs $\G$ and $\H$ such that $\G\to\H$.
Our results imply a classification of problems of this sort in the bipartite regime.
 \begin{corollary}
 Let $\G \to \H$ be two bipartite digraphs with nonempty edge sets. Then $\PCSP(\HG, \HH)$ is tractable if $\G$ is balanced and $\H$ has a directed cycle of length at most 3; otherwise, $\PCSP(\HG, \HH)$ is \NP-hard.
 \end{corollary}
\begin{proof}
    The tractability result follows since when $\G$ has no oriented cycles with nonzero net length and $\H$ has a directed cycle of length at most 3, we have that $\HG \to \zaff \to \HH$. For hardness, note that $\PCSP(\HG, \HH)$ is at least as hard as $\PCSP(\inn{1}{3}, \HH)$ and $\PCSP(\HG, \NAE)$.
    Hence, the hardness follows from Theorems~\ref{thm:main1} and~\ref{thm:main2}.
\end{proof}
\noindent Observe that the tractability boundary in this case is a conjunction of a condition on $\G$ and a condition on $\H$, with these conditions being \emph{independent}. 
Is this a coincidence, or does the independence of the properties drawing the tractability boundary for PCSPs hold in a more general regime?

All
tractable cases of problems $\PCSP(\HG,\HH)$
that we are aware of are those that can sandwich $\zaff$, $\HD_1$, or $\HD_3$ (excluding trivial cases when $\G$ has no edges): $\CSP(\HD_1)$ is trivial, as every instance has a solution; $\CSP(\HD_3)$ corresponds to solving linear equations modulo 3; $\CSP(\Z)$ corresponds to solving linear Diophantine systems.\footnote{Note that sandwiching $\HD_3$ does not imply sandwiching $\zaff$: For an example of a PCSP that sandwiches $\HD_3$, but not $\zaff$, without having any structure in the template homomorphically equivalent to $\HD_3$, let $\X$ be the unique tournament on 4 vertices that has a cycle of length 4, and consider $\PCSP(\HD_9, \widehat{\X})$.} Are all remaining problems in this class \NP-hard?\footnote{Note that the statement is true if $\G$ contains an undirected edge $\D_2$. Indeed, in this case, letting $h$ be the size of $\H$, and $\NAE_h = ([h] ; [h]^3 \setminus \{ (1, 1, 1), \dots, (h, h, h)\})$, as $\D_1 \not \to \H$, we have the sandwich 
$\NAE=\HD_2\to\HG\to\HH\to\NAE_h$,
and $\PCSP(\NAE, \NAE_h)$ is \NP-hard by~\cite{DRS05}.}

\begin{problem}\label{prob:dir1}
    Show that $\PCSP(\HG, \HH)$ is tractable if and only if $(\HG, \HH)$ sandwiches $\zaff, \HD_1$, or $\HD_3$.
\end{problem}

We proved this to be true when $\G$ is balanced (and contains at least one edge), or when $\G$ is an arbitrary unbalanced bipartite digraph and $\H=\D_2$. It is known (see e.g. \cite{notation_old}) and easy to show that $\G\to \D_3$ if and only if all oriented cycles in $G$ have net length divisible by 3. Thus, to answer the above question in the positive, it would be sufficient to prove \NP-hardness for the following two cases: (a) when $\G$ is an oriented cycle of net length not divisible by 3 and $\H$ is a complete graph with at least three vertices, and (b) when $\G$ is an oriented cycle of nonzero net length divisible by 3 and $\H$ is $\D_3$-free.

To see why this is the case, note that when $\G$ has no edges $\PCSP(\HG, \HH)$ is tractable; when $\G$ is balanced and has an edge then $\PCSP(\HG, \HH)$ is tractable if and only if $\PCSP(\inn{1}{3}, \HH)$ is; when all cycles of $\G$ have net length divisible by 3 and $\HD_3 \to \HH$ then $\PCSP(\HG, \HH)$ is tractable by reduction to $\CSP(\HD_3)$. The remaining cases are when $\G$ is unbalanced yet all cycles have net length divisible by 3, and $\H$ is $\D_3$-free, which is hard by reducing from case (b) above, and when $\G$ has an oriented cycle of net length indivisible by 3, which is hard by a reduction from case (a) above. Thus we state the following problem.

\begin{problem}\label{prob:prob2}
    Show that $\PCSP(\HG, \HH)$ is \NP-hard when:
    \begin{enumerate}[label=(\alph*)]
        \item $\G$ is an oriented cycle of net length not divisible by 3 and $\H$ is a complete graph with at least 3 vertices, or
        \item $\G$ is an oriented cycle of net length divisible by $3$ and $\H$ contains no copy of $\D_3$.
    \end{enumerate}
\end{problem}

The previous discussion amounts to the fact that a positive solution to Problem~\ref{prob:prob2} implies a positive solution to Problem~\ref{prob:dir1}.

\medskip

We now discuss direction (ii). It was observed in~\cite{Barto21:stacs,bz22:ic} that assuming the structure $\A$ in a problem $\PCSP(\inn{1}{3},\A)$ to be symmetric does not yield a loss of generality, as, if $\A$ is non-symmetric, replacing it with the maximal symmetric substructure does not alter the complexity of the problem.
In~\cite{Barto21:stacs}, essentially the following question was posed: Are templates of the form $(\inn{1}{3}, \A)$ tractable precisely when $\A$ contains one of the three structures $\HD_1,\HD_2,\HD_3$ as a substructure? When $\A$ is rainbow-free, the current work shows that this is indeed the case.

\medskip

Finally, we observe that extending our results in both directions (i) and (ii) \emph{simultaneously} would amount to classifying the complexity of all problems $\PCSP(\A,\B)$ with $\A$ and $\B$ having a single, ternary relation. Through a similar argument as in~\cite{BG21:sicomp} (see also~\cite{FV98}), this is easily seen to be equivalent to a classification for \emph{all} PCSPs. Such a wide classification appears to be out of reach for the current techniques.

{\small
\bibliographystyle{plainurl}
\bibliography{bib}
}

\end{document}